\documentclass[11pt]{article}%
\usepackage{amsfonts}
\usepackage{amsmath}
\usepackage{amssymb}
\usepackage{graphicx}%
\providecommand{\U}[1]{\protect\rule{.1in}{.1in}}
\def\A{\mathcal A}
\def\P{\mathcal P}
\def\S{\mathcal S}
\def\W{\mathcal W}

\def\RE{\mathbb{R}}

\newtheorem{theorem}{Theorem}

\newtheorem{corollary}[theorem]{Corollary}

\newtheorem{definition}[theorem]{Definition}

\newtheorem{lemma}[theorem]{Lemma}

\newtheorem{remark}[theorem]{Remark}

\newenvironment{proof}[1][Proof]{\noindent\textbf{#1.} }{\ \rule{0.5em}{0.5em}}
\begin{document}

\title{Quantum mappings acting by coordinate transformations on Wigner distributions}
\author{Nuno Costa Dias\textbf{\thanks{Corresponding author}} \textbf{\thanks{ncdias@meo.pt}} \and Jo\~{a}o Nuno
Prata\textbf{\thanks{joao.prata@mail.telepac.pt}}}
\maketitle

\begin{abstract}
We prove two results about Wigner distributions. Firstly, that the
Wigner transform is the only sesquilinear map $\S(\RE^n) \times
\S(\RE^n) \to \S(\RE^{2n})$ which is bounded and covariant under
phase-space translations and linear symplectomorphisms.
Consequently, the Wigner distributions form the only set of
quasidistributions which is invariant under linear symplectic
transformations. Secondly, we prove that the maximal group of
(linear or non-linear) coordinate transformations that preserves
the set of (pure or mixed) Wigner distributions consists of the
translations and the linear symplectic and antisymplectic
transformations.
\end{abstract}
\maketitle

MSC[2010]: {Primary 81S30, 47G30, 53D05; Secondary 47A07, 81P16}

Keywords: {Wigner distribution; quantum mapping, symplectic
covariance; Weyl operator}

\section{Introduction}

In quantum mechanics states are usually represented by {\it
density matrices}. These are positive, trace-class operators $\rho
: L^2(\RE^n) \to L^2(\RE^n)$ with unit trace. The Weyl symbol of
the density matrix operator $\rho$ is the Wigner function
\cite{Wong}:
\begin{equation} \label{eqIntroduction0}
\rho \overset{\mathrm{Weyl}%
}{\longleftrightarrow}W \rho (z) =  \int_{\mathbb{R}^n} K_\rho
\left(x + \frac{y}{2},x - \frac{y}{2} \right) e^{- 2 \pi i  \omega
\cdot y} dy \, ,
\end{equation}
where $K_\rho$ is the Hilbert-Schmidt kernel of $\rho$. The Wigner
function is a familiar quadratic joint representation of position
and momentum of a quantum mechanical state.

Formula (\ref{eqIntroduction0}) can be extended to the non
self-adjoint case: If $\rho$ is a finite rank operator
$\rho_{f,g}$ $(f,g \in L^2(\mathbb{R}^n))$ of the form:
\begin{equation}
\rho_{f,g} h = <h,g>_{L^2}  f, \label{eqIntroduction1} \,
\end{equation}
then the corresponding Wigner function is given by
\cite{Grochenig}:
\begin{equation}
\rho_{f,g}\overset{\mathrm{Weyl}%
}{\longleftrightarrow}W(f, g) (z) =  \int_{\mathbb{R}^n} f \left(x + \frac{y}{2} \right) \overline{g \left(x -
\frac{y}{2} \right)} e^{- 2 \pi i  \omega \cdot y} dy. \label{eqwigner1}
\end{equation}
The Wigner transform $(f,g) \to W(f,g)$ is well defined for all $f \in L^p(\RE^n)$ and $g \in L^{p^{\prime}}(\RE^n)$ with $1 \le p \le \infty$ and $\frac{1}{p}+\frac{1}{p^{\prime}}=1$. Moreover, it can be continuously extended to $f,g \in \S'(\RE^n)$ \cite{Dias2}, in which case $W(f,g) \in \mathcal{S}^{\prime} (\mathbb{R}^{2n})$.

The Wigner function contains the complete information about the quantum state (both in the pure and mixed state
cases). For an arbitrary density matrix $\rho$ and Weyl operator $A$ with Weyl symbol $a \in \mathcal{S}
(\mathbb{R}^{2n})$, we have the following identity:
\begin{equation}
\mbox{tr} \, (A \, \rho)= < a,W\rho >_{L^2 (\mathbb{R}^{2n})}, \label{eqIntroduction2.0}
\end{equation}
In the case $\rho=\rho_{f,g}$ and $f, g \in \mathcal{S}(\RE^n)$, we get:
\begin{equation}
<Af,g>_{L^2 (\mathbb{R}^n)}= <a,W(g,f)>_{L^2 (\mathbb{R}^{2n})}. \label{eqIntroduction2}
\end{equation}

One of the facts that makes the Weyl calculus very popular is that it enjoys the following symplectic covariance
property \cite{Folland,Birk,transam}. If $A: \mathcal{S} (\mathbb{R}^n) \to \mathcal{S}^{\prime} (\mathbb{R}^n)
$ is a Weyl operator with Weyl symbol $a \in \mathcal{S}^{\prime} (\mathbb{R}^{2n}) $, then
\begin{equation}
\widetilde{S}^{-1} A \widetilde{S} \overset{\mathrm{Weyl}%
}{\longleftrightarrow} a \circ S \label{eqIntroduction2.1}
\end{equation}
for any symplectic matrix $S \in Sp(n)$ and any of the two metaplectic operators $\pm \widetilde{S}$ that
project onto $S$. These operators extend to continuous mappings from $\mathcal{S}^{\prime} (\mathbb{R}^n)$ to $\mathcal{S}^{\prime} (\mathbb{R}^n)$.

It follows from (\ref{eqIntroduction2.1}) that if $W(f,g)(z)$ is a
Wigner function then $W(f,g)(Sz)$ is also a Wigner function for
arbitrary $f,g \in L^2 (\mathbb{R}^n)$ and $S \in Sp(n)$.
Moreover, $W(f,g)(Sz)= W(\widetilde{S}^{-1}f, \widetilde{S}^{-1}g)$. Conversely, some heuristic arguments
\cite{Habib} indicate that only the translations and the linear
symplectic and antisymplectic transformations preserve the set of
Wigner functions. In \cite{Dias1} we proved a precise result: if
$M \in Gl(2n, \RE)$ then $W(f,g)(Mz)$ is a Wigner function for all
$f,g \in L^2(\RE^n)$ if and only if $M$ is either a symplectic or
antisymplectic matrix. This result was extended in \cite{JPDOA} to
the case of non-linear coordinate transformations $\phi: \RE^{2n}
\to \RE^{2n}$ belonging to the group $Ham (n)$ of Hamiltonian
symplectomorphisms. It was proved that $W(f,g)(\phi(z)) $ is a
Wigner function for all $f,g \in L^2(\RE^{2n})$ if and only if
$\phi(z)=S z + a$ for some $S \in Sp(n)$ and $a \in \RE^{2n}$.
Notwithstanding the interest of these results, they are still
incomplete: Firstly, they do not apply to the important case of
mixed state Wigner functions. Secondly, and even for pure states,
we still do not know what are the most general coordinate
transformations that preserve the set of Wigner functions.

A generic linear operator mapping a quantum state (density matrix
or Wigner function, pure or mixed) to another state is called -
depending on the context - a quantum map or a positive
trace-preserving map. The characterization of these maps is a
central topic in areas of research like quantum information,
quantum computation, decoherence etc. Two famous results are the
Stinespring theorem and the Kraus theorem. They provide explicit
forms for all completely positive maps
\cite{Keyl,Kraus,Stinespring}. In the case of systems with
continuous variables, the maps that act by coordinate
transformations constitute a sub-class of quantum maps which are
easy to implement experimentally \cite{Braunstein}. They also play
a key role in the definition of separability/entanglement criteria
\cite{Simon} and of quantumness conditions for Gaussian states
\cite{Narcowich1}. Moreover, in the analysis of the semiclassical
limit of quantum mechanics, they are ubiquitous \cite{Littlejohn}.
Outside from quantum mechanics, nonlinear symplectic
transformations are believed to characterize aberration effects in
the wave and ray theory of light \cite{Dragt}. Also, a certain
nonlinear coordinate transformation was used to approximate the
propagation of the Wigner distribution of a pulse in a general
dispersive medium \cite{Loughlin}.

This paper is divided in two parts. In the first part we will
prove a uniqueness result about the covariance properties of the
Wigner transform. In the second part we will determine {\it all}
quantum maps that act by coordinate transformations on {\it all}
main sets of Wigner distributions. More precisely:

\vspace{0.3 cm} \noindent {\bf (I)} It is a well-documented fact that all sesquilinear maps from $L^2
(\mathbb{R}^n) \times L^2 (\mathbb{R}^n)$ to the set of measurable functions on $\mathbb{R}^{2n}$ which are
covariant under time-frequency translations belong to the so-called Cohen class \cite{Cohen}. In Theorem
\ref{Theorem1} we prove that if we also add the requirement of covariance under linear symplectomorphisms then
the Wigner transform is the {\it unique} solution. Hence, the set of Wigner distributions is the only set of
quasidistributions which is invariant under linear symplectic transformations.

This seems to be an expected result that could presumably be proven by imposing the symplectic covariance property
directly on the Cohen class of quasidistributions \cite{Cohen}. However, as we will see in section 4, this
approach does not pin down the Wigner transform uniquely in an obvious way, which might be the reason why, up to our
knowledge, this result has never been presented in the literature. In section 3, we will use a different
approach and prove the uniqueness result in an elegant way, directly from the properties of the metaplectic
group.

The question of identifying conditions that determine uniquely the
phase-space representative of a quantum mechanical state have been considered previously. In \cite{OconnellWigner} O'Connell and Wigner stated a certain number of conditions which determine the
Wigner function uniquely. Compared with our result, they impose
positive marginal distributions and Moyal's identity, whereas we
require symplectic invariance. Of course, the conditions required
depend on the context and the application that one has in mind. If
one is more interested in the probabilistic interpretation of
quantum mechanics or the energy content of a signal in signal
processing, then conditions such as proposed by O'Connell and
Wigner seem to be more appropriate. If one is more interested in
symmetry questions such as symplectic invariance which appear in
the semiclassical limit of quantum mechanics, in quantum
information theory or in quantum optics, then our conditions are
more natural.

\vspace{0.3 cm} \noindent {\bf (II)} In the second part of the
paper (section 5) we determine all coordinate transformations that
leave the sets of pure, mixed and distributional Wigner
distributions invariant. These results extend the results of
\cite{Dias1,JPDOA} in two different directions: i) The coordinate
transformations are not {\it a priori} restricted to a specific
set (in \cite{Dias1} they were assumed to be linear, and in
\cite{JPDOA} only Hamiltonian symplectomorphisms were considered)
and ii) the results are valid for all main sets of Wigner
distributions (and not only for pure states). Most significative
is Theorem \ref{TheoremPurity2}, where a complete result is proven
for the set of mixed states.

To state our results precisely, let us define the following sets
of Wigner distributions: Let $\W^2$ be the range of the transform
(\ref{eqwigner1}) for $f,g \in L^2(\RE^n)$. The subset of $\W^2$
which consists of the diagonal elements $W(f,f)$ with $||f||_{L^2 (\mathbb{R}^n)}=1$, is denoted by
$\W_+^2$. The elements of $\W_+^2$ are the true quantum mechanical
pure states, but nondiagonal elements of the form
(\ref{eqwigner1}) appear frequently in quantum mechanics, when one
considers linear combinations of wave functions. Let also
$\W^{\prime}$ be the range of (\ref{eqwigner1}) for $f,g \in
\S'(\RE^n)$. Finally, let $\W_M$ be the set of Weyl symbols of the
(pure and mixed) density matrices (i.e. positive, trace-class
operators $\rho: L^2(\RE^n) \to L^2(\RE^n)$ with unit trace). We
have, of course, $\W_+^2 \subset \W_M$ and $\W_+^2 \subset \W^2
\subset \W^{\prime}$.

Let us also consider the following sets of linear maps acting by
coordinate transformations:

\begin{definition} \label{DefinitionCT}

Let $\A=\W^2, \W_+^2, \W^{\prime}, \W_M$. Then ${\mathcal U}_\A$
is the set of all linear maps $U_{\phi}: \A \to \S'(\RE^{2n}) $ defined by:
\begin{equation}
\left(U_{\phi} F \right) (z) :
= J(z) F (\phi (z)) \label{eq2}
\end{equation}
where $\phi: \mathbb{R}^{2n} \to \mathbb{R}^{2n}$ is a $C^1$ diffeormorphism with Jacobian
\begin{equation}
J(z) = \left|\det \left( \frac{\partial \phi_i}{\partial z_j} \right)_{1 \le i,j \le 2n} \right| .
\label{eq2.1}
\end{equation}

\end{definition}

We remark that the Jacobian is included in (\ref{eq2}) for the sake of preserving the normalization:
\begin{equation}
\int_{\mathbb{R}^{2n}} \left(U_{\phi} F  \right) (z) dz =
\int_{\RE^{2n}} F (z) dz \quad , \quad \forall F \in L^1(\RE^{2n})
\label{eq3}
\end{equation}
Notice that the Jacobian is not required to be everywhere strictly positive definite.

We also remark that in the three cases $\A=\W^2,\W_+^2,\W_M$, the
requirement $U_{\phi} F \in \A$ immediately implies that $\phi \in
C^1$ (because then $U_{\phi} F$ has to be uniformly continuous
\cite{Grochenig}). This is also the case for $\A=\W^{\prime}$ (see
the proof of Corollary \ref{CorollaryCovariance}).

We now notice that, in general, it is not true that $U_{\phi}F \in
\A$ for all $F \in \A$. We will show in Theorem
\ref{TheoremCovariance}, Corollary \ref{CorollaryCovariance} and
Theorem \ref{TheoremPurity2} that, in all four cases
$\A=\W^2,\,\W_+^2,\, \W^{\prime}, \W_M$, the map $U_{\phi}$ is an
inner operation in $\A$ if and only if $\phi(z)=Mz + a$ with $a
\in \RE^{2n}$ and $M$ a symplectic or antisymplectic matrix. As a
byproduct of theses results, we argue in Remarks \ref{Remark1} and
\ref{Remark2} that the same conclusion is valid for maps of the
form $U_{\phi}: \W_+^2 \to \W^{\prime}$ and $U_{\phi}: \W_+^2 \to
\W_M$.

\vspace{0.3 cm} \noindent {\bf (III)} In the appendix we prove a
simple result about polynomials which is used in the proof of
Theorem \ref{TheoremCovariance}. We include this result in the
paper for completeness, because we were unable to find it in the
literature. It is a side result that, nevertheless, looks
interesting: it determines all real and continuous functions $f,g$
such that $f^2,g^2$ and $fg$ are all second order polynomials.

\vspace{0.3 cm}\noindent\textbf{Notation 1.} The inner product in
$\mathbb{R}^n$ is $u \cdot v = \sum_{i=1}^n u_i v_i$ for $u=(u_1,
\cdots, u_n),v=(v_1, \cdots, v_n) \in \mathbb{R}^n$ and $|u|^2 = u
\cdot u$. The space of continuous functions on $\mathbb{R}^n$ is
denoted by $C(\mathbb{R}^n)$. The Schwartz class of test functions
is written $\mathcal{S} (\mathbb{R}^n)$ and its dual
$\mathcal{S}^{\prime} (\mathbb{R}^n)$ is the space of tempered
distributions.

\vspace{0.3 cm}\noindent\textbf{Notation 2.} The standard
symplectic form on $\mathbb{R}^n \oplus \mathbb{R}^n$ is given by:
\begin{equation}
\left[z, z^{\prime} \right] = - z^T \mathcal{J} z^{\prime} = x \cdot \omega^{\prime}- x^{\prime} \cdot \omega ,
\label{eqIntroduction3}
\end{equation}
where $z=(x, \omega)$, $z^{\prime}=(x^{\prime}, \omega^{\prime})$ and
\begin{equation}
\mathcal{J} = \left(
\begin{array}{c c}
0 & I\\
-I & 0
\end{array}
\right) \label{eqIntrosuction4}
\end{equation}
is the standard symplectic matrix. We recall that $M$ is symplectic if $\left[M z, M z^{\prime} \right] =
\left[z, z^{\prime} \right]$, and anti-symplectic if $\left[M z, M z^{\prime}
 \right] = - \left[z, z^{\prime} \right]$, for all $z,z^{\prime} \in \mathbb{R}^{2n}$.

Anti-symplectic transformations amount to a symplectic transformation followed by a "time"-reversal:
\begin{equation}
z= \left(
\begin{array}{c}
x\\
\omega
\end{array}
\right) \mapsto Tz= \left(
\begin{array}{c}
x\\
- \omega
\end{array}
\right)
\label{eqTimereversal}
\end{equation}
This is interpreted as a time-reversal since it reverses a particle's momentum. We denote by $Sp(n)$ the symplectic group of real $2n \times 2n$
symplectic matrices and by $SpT(n) = Sp (n) \cup \left\{ T \right\}$ the group of all real $2n
\times 2n$ matrices which are either symplectic or antisymplectic.

Moreover, $Sym (n; \mathbb{R})$ is the set of real symmetric $n
\times n$ matrices. The set of symplectic matrices which are also
symmetric and positive-definite is denoted by $Sp^+ (n)$. Finally,
$\textsf{sp} (n)$ is the symplectic algebra.

\vspace{0.3 cm}\noindent\textbf{Notation 3.} The
Fourier-Plancherel transform of a function $f \in L^1
(\mathbb{R}^n) \cap L^2 (\mathbb{R}^n)$ is given by:
\begin{equation}
(\mathcal{F} f) (\omega) = \widehat{f} (\omega):= \int_{\mathbb{R}^n} f (x) e^{- 2 \pi i x \cdot \omega} d x,
\label{eqFourierTransform}
\end{equation}
which extends to an isometry in $L^2 (\mathbb{R}^n) $. From the point of view of quantum mechanics, it amounts to setting
the Planck constant $h=1$ or the more familiar $\hbar:=\frac{h}{2
\pi}= \frac{1}{2 \pi}$. In all subsequent formulae Planck's
constant can be recovered by a simple dilation $z \mapsto
\frac{z}{\sqrt{h}}$.

\section{Preliminaries}

In this section, we recapitulate some definitions and results which will be needed in the sequel.

The metaplectic group $Mp(n)$ is a unitary representation of the two-fold cover $Sp_2(n)$ of $Sp(n)$. We
denote by $\pi$ the projection from $Mp(n)$ onto $Sp(n) \cong
Mp(n) / \left\{ \pm I \right\}$. We then have $\pm \widetilde{S}
\mapsto \pi (\pm \widetilde{S} ) =S$.

The fundamental operators in Weyl quantization are the Heisenberg-Weyl operators defined by:
\begin{equation}
(\rho (z_0) f ) (x) = e^{2 \pi i \omega_0 \cdot \left(x-
\frac{x_0}{2} \right)} f (x-x_0) \label{eq3}
\end{equation}
for $f \in \mathcal{S} (\mathbb{R}^n)$, and $z_0 =(x_0, \omega_0)$. They extend to unitary operators in $L^2 (\mathbb{R}^n)$. An operator of the form
\begin{equation}
\rho(z, \tau) = e^{2 \pi i \tau} \rho (z), \label{eq5.1}
\end{equation}
for $(z, \tau) \in \mathbb{R}^{2n} \times \mathbb{R}$ is known as the Schr\"odinger representation of the
Heisenberg group $\mathbb{H} (n)$ (see e.g. \cite{Folland,Grochenig}).

\begin{theorem}\label{TheoremShaleWeil}
{\bf (Shale-Weil relation)} Let $S \in Sp(n)$ and $\pm \widetilde{S} \in Mp(n)$ the two metaplectic operators projecting onto $S$. Then
\begin{equation}
\widetilde{S} \rho (z) \widetilde{S}^{-1} = \rho (S z), \label{eqShaleWeil1}.
\end{equation}
\end{theorem}

A proof of the previous theorem can be found in \cite{Folland,Birk,Shale,Weil}.

The Wigner transform of the pair $(f,g) \in L^p (\mathbb{R}^n) \times L^{p^{\prime}} (\mathbb{R}^n) $ is defined by way of eq.(\ref{eqwigner1}). H\"older's inequality guarantees the continuity of $W(f,g)$. If $f=g$, then we simply write $W f$, meaning $W(f,f)$. In quantum mechanics $W f$ is interpreted as the quasi-probability density associated with the state $f \in L^2 (\mathbb{R}^n)$, whenever $||f||_2=1$. Moyal's identity states that \cite{Grochenig}
\begin{equation}
<W(f_1,f_2), W (g_1,g_2) >_{L^2 (\mathbb{R}^{2n})} = <f_1, g_1 >_{L^2 (\mathbb{R}^n)} <g_2, f_1 >_{L^2 (\mathbb{R}^n)}
\label{eqMoyalIdentity1}
\end{equation}
which entails, in particular that
\begin{equation}
<W f, Wg> _{L^2 (\mathbb{R}^{2n})} = \left|<f, g >_{L^2 (\mathbb{R}^n)} \right|^2 \ge 0 ,
\label{eqMoyalIdentity2}
\end{equation}
while
\begin{equation}
||W f||_{L^2 (\mathbb{R}^{2n})}^2 = ||f||_{L^2 (\mathbb{R}^n)}^4=1.
\label{eqMoyalIdentity3}
\end{equation}

Alternatively, we can regard the finite rank operator $\rho_f:= \rho_{f,f}$ (see (\ref{eqIntroduction1})) as a Weyl operator
\begin{equation}
(\rho_f h)(x) = \int_{\mathbb{R}^n} K_f(x,y) h(y) dy,
\label{eqDensityMatrix1}
\end{equation}
with kernel
\begin{equation}
K_f(x,y) = (f \otimes \overline{f}) (x,y) = f(x) \overline{f (y)}
\label{eqDensityMatrix2}
\end{equation}
and Weyl symbol given by the Wigner function $Wf$:
\begin{equation}
 \int_{\mathbb{R}^n} K_f \left(x + \frac{y}{2},x -
\frac{y}{2} \right) e^{- 2 \pi i  \omega \cdot y} dy= \int_{\mathbb{R}^n} f \left(x + \frac{y}{2} \right) \overline{f \left(x -
\frac{y}{2} \right)} e^{- 2 \pi i  \omega \cdot y} dy
\label{eqDensityMatrix2.1}
\end{equation}

In quantum mechanics, statistical mixtures appear naturally in most experimental setups, so that one generally ends up with convex combinations of Wigner functions of the form:
\begin{equation}
\sum_{\alpha} p_{\alpha} W f_{\alpha} (z),
\label{eqDensityMatrix3}
\end{equation}
where
\begin{equation}
0 \le p_{\alpha}, \hspace{1 cm} \sum_{\alpha} p_{\alpha}  =1
\label{eqDensityMatrix4}
\end{equation}
and the index $\alpha$ takes values in some subset of $\mathbb{N}$. The sum in (\ref{eqDensityMatrix3}) is the
Weyl symbol
\begin{equation}
 \int_{\mathbb{R}^n} K_{\rho} \left(x + \frac{y}{2},x -
\frac{y}{2} \right) e^{- 2 \pi i  \omega \cdot y} dy
\label{eqDensityMatrix5}
\end{equation}
of a Weyl operator acting on $L^2 (\mathbb{R}^n)$
\begin{equation}
(\rho h) (x) = \int_{\mathbb{R}^n} K_{\rho} (x,y) h(y) dy
\label{eqDensityMatrix5}
\end{equation}
with a Hilbert-Schmidt kernel $K_{\rho} =\sum_{\alpha} p_{\alpha} K_{f_{\alpha}} $. This operator is of the form
\begin{equation}
\rho = \sum_{\alpha} p_{\alpha} \rho_{f_{\alpha}}.
\label{eqDensityMatrix5.1}
\end{equation}
The previous series converges in the trace-norm. It can be shown that such operators are positive
and trace-class with unit trace. They are commonly known as
density matrices. Those (such as that in
eq.(\ref{eqDensityMatrix1})) which satisfy
\begin{equation}
\rho^2 = \rho
\label{eqDensityMatrix6}
\end{equation}
are said to represent pure states, otherwise the states are called mixed. If $W \rho$ denotes the Wigner function of a density matrix $\rho$, then (cf.(\ref{eqMoyalIdentity3}))
\begin{equation}
Tr \rho =1 \Leftrightarrow ||W \rho||_{L^2 (\mathbb{R}^{2n})} =1
\label{eqDensityMatrix6}
\end{equation}
if the state is pure, otherwise
\begin{equation}
Tr \rho < 1 \Leftrightarrow ||W \rho||_{L^2 (\mathbb{R}^{2n})} <1
\label{eqDensityMatrix6}
\end{equation}
for a mixed state. For this reason, one usually calls $||W \rho||_{L^2 (\mathbb{R}^{2n})}^2$ the purity of the state.

A difficult problem consists of determining whether a given
function $F(z)$ on the phase-space is the Wigner function
associated with some density matrix $\rho$. In other words, how
can one tell whether $F \in \W_M$? The answer is stated in the
following theorem (see e.g.\cite{Dias3,Dias4,Lions,Narcowich}):

\begin{theorem}\label{TheoremWignerFunctions}
Let $F: \mathbb{R}^{2n} \to \mathbb{C}$. Then $F \in \W_M$, if and
only if the following conditions hold true:

\vspace{0.3 cm} \noindent (i) $F \in L^2 (\mathbb{R}^{2n}) \cap C
(\mathbb{R}^{2n})$,

\vspace{0.3 cm}
\noindent
(ii) $\overline{F (z)} =F(z)$ everywhere in $\mathbb{R}^{2n}$,

\vspace{0.3 cm}
\noindent
(iii) $\int_{\mathbb{R}^{2n}} F(z) dz =1$, and

\vspace{0.3 cm}
\noindent
(iv) $\int_{\mathbb{R}^{2n}} F(z) W f (z) dz \ge 0$ for all $f \in L^2 (\mathbb{R}^n)$.
\end{theorem}

For future reference, we state the following Lemma.

\begin{lemma}\label{LemmaPurity1}
Let $F \in \W_M \cup \W_+^2$. Then
\begin{equation}
||W F||_{L^{\infty} (\mathbb{R}^{2n})} \le 2^n.
\label{eqPurity1}
\end{equation}
Moreover, for any $z_0 \in \mathbb{R}^{2n}$, there exists $f \in L^2 (\mathbb{R}^{n})$ such that
\begin{equation}
|W f (z_0)|=2^n.
\label{eqPurity2}
\end{equation}
\end{lemma}

\begin{proof}
This is a well known result (see for instance \cite{Royer} for pure states) which
can be proven by a simple application of the Cauchy-Schwartz
inequality. For mixed states the proof follows by uniform convergence of the series (\ref{eqDensityMatrix3},\ref{eqDensityMatrix4}).
\end{proof}

Wigner measures cannot be regarded as joint probability measures
for position and momentum (or time and frequency) as they may not
be positive measures. Gaussian and pseudo-Gaussians are some of
the exceptions.

\begin{definition}\label{definitionHudson1}
The tempered distribution $f$ on $\mathbb{R}^n$ is called a pseudo-Gaussian if either (i) there is a proper subspace $V \subset \mathbb{R}^n$ and an element $y_0 \in V^{\perp}$ such that $f(x_1,x_2) = f_0 (x_1) \delta_{y_0} (x_2)$, where $x_1 \in V$, $x_2 \in V^{\perp}$, $f_0 (x_1) =C e^{- Q(x_1)}$ for some complex number $C$ and $Q$ is a polynomial of degree $2$ such that $Re Q$ is bounded from below, and $\delta_{y_0}$ is a Dirac measure, or (ii) $f(x) = f_0 (x) $, where $x \in \mathbb{R}^n$, $f_0 (x) =C e^{- Q(x)}$ for some complex number $C$ and $Q$ is a polynomial of degree $2$ such that $Re Q$ is bounded from below. In the latter case, $f$ is also called a semi-Gaussian. Notice that $f$ is a Gaussian if and only if it is a semi-Gaussian and $f(x) \to 0$ as $|x| \to \infty$.
\end{definition}

The following theorem in commonly known as Hudson's Theorem \cite{Grochenig,Hudson,Janssen,Toft}.

\begin{theorem}\label{HudsonTheorem}
{\bf (Hudson's Theorem)} Assume $f,g \in \mathcal{S}^{\prime} (\mathbb{R}^n) \backslash \left\{0 \right\} $. Then

\vspace{0.3 cm}
\noindent
1) $W(f,g)$ is a positive measure, if and only if $f$ is a pseudo-Gaussian and $g=c f$ for some constant $c>0$;

\vspace{0.3 cm}
\noindent
2) If, in addition, $f,g \in L^p (\mathbb{R}^n)$ for some $1 \le p < \infty$, then $W(f,g)$ is a positive measure, if and only if $f$ is a Gaussian and $g=cf$ for some constant $c>0$.
\end{theorem}

\begin{proof}
For a proof see e.g. \cite{Toft}.
\end{proof}

\section{The uniqueness of the Wigner transform}

There are a number of properties that seem natural to require of a phase space representation of (non-diagonal)
density matrix elements of the form $\rho_{f,g}$ (see (\ref{eqIntroduction1})). These correspond to Weyl
operators with kernel $K_{\rho_{f,g}}= f \otimes \overline{g}$. Since this quantity is sesquilinear, so should
be its phase space counterpart $\mathcal{Q} (f,g)$. If the wave functions $f$ and $g$ undergo a translation by
an amount $x_0$ and their Fourier transforms $\widehat{f}$, $\widehat{g}$ increase the momentum by an amount
$\omega_0$, then $\mathcal{Q} (f,g)(z)$ should be translated to $\mathcal{Q} (f,g) (z-z_0)$ in the phase space,
where $z_0=(x_0,\omega_0)$. Equally, if a metaplectic transformation $\widetilde{S}$ acts on $f$ and $g$,
then the effect on the phase space representation ought to be $\mathcal{Q} (f,g)(S^{-1} z)$. Moreover, the
diagonal elements $\rho_f \def \rho_{f,f}$ must have unit trace. It then seems natural to require that the
integral of $\mathcal{Q} (f,f) (z)$ be finite. These conditions are stated explicitly in the following theorem,
which shows that the Wigner transform is the only representative of $\rho_{f,g}$ which satisfies this set of
conditions.

\begin{theorem}\label{Theorem1}
Let $\mathcal{Q}$ be a mapping $L^2 (\mathbb{R}^n) \times L^2 (\mathbb{R}^n) \to L^{\infty} (\mathbb{R}^{2n})
$. Then $\mathcal{Q}$ is proportional to the Wigner transform, if and only if:

\vspace{0.3 cm} \noindent {\bf 1)} $\mathcal{Q}(\alpha_1 f_1+ \alpha_2 f_2, g) = \alpha_1 \mathcal{Q} (f_1, g) +
\alpha_2 \mathcal{Q}(f_2, g)$;

\vspace{0.3 cm} \noindent {\bf 2)} $\mathcal{Q}(f, \beta_1 g_1+ \beta_2 g_2) = \overline{\beta_1} \mathcal{Q}
(f, g_1) + \overline{\beta_2} \mathcal{Q}(f, g_2)$;

\vspace{0.3 cm} \noindent {\bf 3)} There exists $C>0$ such that $|\mathcal{Q}(f, g) (0)| \le C || f||_2 ||
g||_2$;

\vspace{0.3 cm} \noindent {\bf 4)} Given $z_0 \in \mathbb{R}^{2n}$, then $\mathcal{Q}(f, g) (z -z_0) =
\mathcal{Q} (\rho (z_0) f, \rho (z_0) g) (z)$ for all  $z \in \mathbb{R}^{2n}$, where $\rho (z_0)$ is the
Heisenberg-Weyl operator;

\vspace{0.3 cm} \noindent {\bf 5)} There exists $f_0 \in L^2 (\mathbb{R}^n)$ such that $\int_{\mathbb{R}^{2n}}
\mathcal{Q} (f_0, f_0) (z) dz$ is finite;

\vspace{0.3 cm} \noindent {\bf 6)} For any $\widetilde{S} \in Mp (n)$, $\mathcal{Q}(\widetilde{S} f,
\widetilde{S}g) (0) = \mathcal{Q}( f , g) (0)$,

\vspace{0.3 cm} \noindent for arbitrary $\alpha_1, \alpha_2, \beta_1, \beta_2 \in \mathbb{C}$ and $f_1, f_2, f,
g_1, g_2, g \in L^2 (\mathbb{R}^n)$.

\end{theorem}

\begin{proof}
That the Wigner transform satisfies all the previous conditions is a well known fact.

Conversely, suppose that the mapping $\mathcal{Q}$ satisfies all the above requirements. From {\bf 1)}, {\bf 2)}
and {\bf 3)}, we conclude that the mapping $(f, g) \mapsto \mathcal{Q}(f, g) (0)$ is a bounded sesquilinear
form. By the Riesz representation theorem, there exists a bounded linear operator $U (0)$ such that
\begin{equation}
\mathcal{Q}(f, g) (0) = < U (0) f, g>_{L^2}. \label{equniqueness1}
\end{equation}
From the previous equation and {\bf 4)}, we obtain:
\begin{equation}
\begin{array}{c}
\mathcal{Q} (f,g) (z) = \mathcal{Q} (\rho (-z) f, \rho (-z) g) (0)=\\
\\
=<U(0) \rho (-z) f, \rho (-z) g >_{L^2}= <\rho (z) U(0) \rho (z)^{-1} f,  g >_{L^2},
\end{array}
\label{equniqueness2}
\end{equation}
for all $f,g  \in L^2 (\mathbb{R}^{n})$, and where we used the fact that $\rho(z)^{\ast} = \rho(z)^{-1} =
\rho(-z)$.

Next we define:
\begin{equation}
U(z)  = \rho(z) U (0) \rho(z)^{-1}. \label{eq11}
\end{equation}

Notice that from (\ref{equniqueness1}) and {\bf 6)} we conclude that $U (0)$ commutes with $\widetilde{S} $ for all
$\widetilde{S} \in Mp (n)$:
\begin{equation}
\widetilde{S} U(0) \widetilde{S}^{-1} = U (0). \label{eq12}
\end{equation}
Consequently, we have the desired symplectic covariance property:
\begin{equation}
\widetilde{S} U (z) \widetilde{S}^{-1} = U (S z) \Rightarrow \mathcal{Q} (\widetilde{S}^{-1} f,
\widetilde{S}^{-1} g ) (z) = \mathcal{Q}(f, g) (Sz), \label{eq12}
\end{equation}
where we used the Shale-Weil relation. Finally, we recall that the metaplectic representation is not irreducible.
Indeed, it has two non-trivial invariant subspaces, the subspaces of even and odd functions \cite{Folland}.
Since $U (0)$ commutes with $\widetilde{S}$ for all $\widetilde{S}  \in Mp (n)$, in these subspaces it must
be proportional to the identity operator according to Schur's Lemma.

Let $I$ denote the identity operator and $R$ the reflection operator $(R f ) (x) = f (-x)$ in $L^2
(\mathbb{R}^n)$. The projections on the subspaces of even $(+)$ and odd $(-)$ functions are $P_{\pm} =
\frac{1}{2} (I \pm R)$. From the previous analysis we conclude that:
\begin{equation}
U (0) P_+ = \alpha P_+, \hspace{1 cm} U(0) P_- = \beta P_-, \label{eq13}
\end{equation}
for some constants $\alpha, \beta \in \mathbb{C}$. The solution to both equations is:
\begin{equation}
U(0)  = \alpha  P_+ + \beta  P_- = \frac{\alpha + \beta}{2} I + \frac{\alpha - \beta}{2} R  . \label{equniqueness7}
\end{equation}
Altogether, from (\ref{equniqueness2},\ref{equniqueness7}) we obtain:
\begin{equation}
\mathcal{Q} (f, g) (z) = \frac{\alpha + \beta}{2} <f,g>_{L^2} + \frac{\alpha - \beta}{2} <G (z) f,g>_{L^2},
\label{eq15}
\end{equation}
where
\begin{equation}
G (z)  = \rho (z)   R \rho(z)^{-1} \label{eq16}
\end{equation}
is the Grossmann-Royer operator \cite{Birk,Grossmann,Royer}:
\begin{equation}
(G(z_0)f) (x) = e^{4 \pi i \omega_0 \cdot(x-x_0)} f (2x_0-x), \label{eq16.1}
\end{equation}
for $z_0=(x_0, \omega_0)$.

Finally, from {\bf 5)} we conclude that $\beta=  - \alpha$ and thus
\begin{equation}
\mathcal{Q}(f,g) (z) = \alpha <G (z) f,g>_{L^2}. \label{eq17}
\end{equation}
On the other hand, it is well known \cite{Birk,Grossmann,Royer} that the cross-Wigner function can be written as
\begin{equation}
W(f, g) (z) = 2^n <G (z) f,g>_{L^2}, \label{eq18}
\end{equation}
and this concludes the proof.
\end{proof}

\section{Remarks on Cohen's class}

Notice that the condition {\bf 6)} which leads to the symplectic covariance is crucial for uniqueness in the
previous theorem. Indeed, there are many quadratic representations which satisfy all the previous conditions
except {\bf 6)}.

It is a well known fact \cite{Grochenig} that if a mapping
$\mathcal{Q}: \mathcal{S} (\mathbb{R}^n) \times \mathcal{S}
(\mathbb{R}^n) \to \mathcal{S} (\mathbb{R}^{2n})$ satisfies
conditions {\bf 1)-4)} (and {\bf 5)} by construction), then there
exists $\sigma \in \mathcal{S}^{\prime} (\mathbb{R}^{2n})$, such
that
\begin{equation}
\mathcal{Q}(f,g) (z)= \mathcal{Q}_{\sigma } (f, g) (z) = (\sigma \star W (f, g) ) (z) = \int_{\mathbb{R}^{2n}}
\sigma (z- z^{\prime}) W(f, g) (z^{\prime}) d z^{\prime}. \label{eq19}
\end{equation}
A quadratic representation of this form is said to belong to
Cohen's class \cite{Cohen}. It is clear, however, that not all
maps (\ref{eq19}) (with $\sigma \in
\mathcal{S}'(\mathbb{R}^{2n})$) satisfy the conditions {\bf 1)-5)}
(take, for instance, $\sigma=1$). However, for large classes of
distributions $\sigma$ the maps (\ref{eq19}) do satisfy {\bf
1)-5)}. Let us illustrate this point by taking $\sigma \in L^1
(\mathbb{R}^{2n}) \cap C (\mathbb{R}^{2n})$.

That $\mathcal{Q}_{\sigma } (f, g) $ is sesquilinear in $f$ and $g$ is obvious.

Next:
\begin{equation}
\begin{array}{c}
|\mathcal{Q}_{\sigma } (f, g)  (0)| \le \int_{\mathbb{R}^{2n}} |\sigma (- z^{\prime}) | ~ |W(f, g) (z^{\prime})| d z^{\prime} \le\\
\\
\le C ||f||_2 ||g||_2 ||\sigma ||_1.
\end{array}
\label{eq20}
\end{equation}

Moreover, if $f  \in \mathcal{S} (\mathbb{R}^n)$, then from Fubini's Theorem:
\begin{equation}
|\int_{\mathbb{R}^{2n}} \mathcal{Q}_{\sigma} (f,f) (z) dz | \le ||\sigma||_1 ~|| W f||_1, \label{eq21}
\end{equation}
which means that $\int_{\mathbb{R}^{2n}} \mathcal{Q}_{\sigma} (f,f) (z) dz $ is finite for all $f  \in
\mathcal{S} (\mathbb{R}^n)$. This proves {\bf 5)}.

Also:
\begin{equation}
\begin{array}{l}
 \mathcal{Q}_{\sigma} (f,g) (z -z_0 ) = \int_{\mathbb{R}^{2n}} \sigma (z-z_0- z^{\prime}) W(f, g) (z^{\prime}) d z^{\prime} =\\
 \\
 = \int_{\mathbb{R}^{2n}} \sigma (z- z^{\prime \prime}) W(f, g) (z^{\prime \prime} -z_0) d z^{\prime \prime} =\\
 \\
 = \int_{\mathbb{R}^{2n}} \sigma (z- z^{\prime \prime}) W(\rho (z_0)f, \rho (z_0) g) (z^{\prime \prime} ) d z^{\prime \prime} =\\
 \\
 = \mathcal{Q}_{\sigma} (\rho (z_0)f,\rho (z_0)g) (z).
\end{array}
\label{eq22}
\end{equation}
So, we conclude that if $\sigma \in  L^1 (\mathbb{R}^{2n}) \cap C
(\mathbb{R}^{2n})$, then $\mathcal{Q}_{\sigma} $ automatically
satisfies all the conditions {\bf 1)-5)} of Theorem
\ref{Theorem1}.

However, none of the maps (\ref{eq19}) with $\sigma \in L^1
(\mathbb{R}^{2n}) \cap C (\mathbb{R}^{2n})$ satisfies {\bf 6)}. To
prove this let $\widetilde{S} \in Mp(n)$ and consider the
relation:
\begin{equation}
\begin{array}{l}
\mathcal{Q}_{\sigma} (\widetilde{S}f,\widetilde{S}g) (0) = \int_{\mathbb{R}^{2n}} \sigma (- z^{\prime}) W(\widetilde{S}f,\widetilde{S}g) (z^{\prime}) d z^{\prime} =\\
\\
=\int_{\mathbb{R}^{2n}} \sigma (- z^{\prime}) W(f, g) (S^{-1} z^{\prime}) d z^{\prime} =\\
\\
=\int_{\mathbb{R}^{2n}} \sigma (- S z^{\prime \prime}) W(f, g) ( z^{\prime \prime}) d z^{\prime \prime} =
\mathcal{Q}_{\sigma \circ (-s)} (f,g) (0)
\end{array}
\label{eq23}
\end{equation}
where $\widetilde{S} $ projects on $S \in Sp(n)$ and $s$ is the symplectic automorphism $s(z) =Sz$. If we impose
{\bf 6)}, then
\begin{equation}
\mathcal{Q}_{\sigma \circ (-s)} (f,g) (0) = \mathcal{Q}_{\sigma } (f,g) (0) \label{eq24}
\end{equation}
for all linear symplectomorphisms $s$ and all $f, g \in \mathcal{S} (\mathbb{R}^n)$. This is equivalent to:
\begin{equation}
\sigma (Sz) = \sigma(z) \label{eq25}
\end{equation}
for all $z \in \mathbb{R}^{2n}$ and all $S \in Sp(n)$. Since $\sigma$ is continuous and $Sp(n)$ acts
transitively on $\mathbb{R}^{2n} \backslash \left\{ 0  \right\}$ \cite{Moskowitz}, this is only possible if
$\sigma$ is a constant function. But that contradicts $\sigma \in L^1 (\mathbb{R}^{2n})$.

One may expect that by imposing the symplectic covariance over the
complete Cohen class (\ref{eq19}), we could obtain an alternative
(maybe simpler) proof of Theorem \ref{Theorem1} (roughly as
follows: conditions {\bf 1)-5)} imply that $\mathcal{Q}$ is of the
form (\ref{eq19}) and condition {\bf 6)} further implies that
$\sigma$ has to be the Dirac measure $\delta_z$). Unfortunately,
this is not so simple because conditions {\bf 1)-4)} also impose
some extra restrictions on the set of admissible distributions
$\sigma$, which are not easily described, but are decisive to pin
down the Wigner transform. The symplectic covariance alone,
imposed on the mapping (\ref{eq19}), leads to a condition for
$\sigma$ which is just the distributional generalization of
equation (\ref{eq25}) (which follows directly from the
distributional generalizations of equations
(\ref{eq23},\ref{eq24})). This condition is not sufficient to
select the unique solution $\sigma=\delta_z$, as one easily
realizes. In fact, all distributions $\Delta$ such that supp
$\Delta=\{0\}$ are also solutions of (\ref{eq25}). It may be
possible to remove these extra solutions by a careful
consideration of conditions {\bf 1)-4)}. However, our proof of
Theorem \ref{Theorem1} provides a more direct construction of the
uniqueness result.

\section{Covariance group of Wigner distributions under coordinate transformations}

In this section we prove the results summarized in the point II of
the Introduction. Most significative are Theorem
\ref{TheoremCovariance} (the pure state case) and Theorem
\ref{TheoremPurity2} (the mixed state case). The latter theorem ia
a definitive result about the characterization of quantum mappings
acting by coordinate transformations (diffeomorphisms).

Let us start with the following preparatory Lemma.

\begin{lemma}\label{LemmaJacobian}
Let $\mathcal{A}$ be one of the sets of Wigner distributions
$\W^2, \W_+^2$ or $\W_M$. Let the map $U_{\phi} \in
\mathcal{U}_{\A}$ be given by (\ref{eq2}). If $U_{\phi}$ is of the
form
\begin{equation}
U_{\phi} : \mathcal{A} \to \mathcal{A}, \label{eqJacobian1}
\end{equation}
then
\begin{equation}
J(z) \le 1,
\label{eqJacobian2}
\end{equation}
for all $z \in \mathbb{R}^{2n}$.
\end{lemma}

\begin{proof}
Suppose, there exists $z_1 \in \mathbb{R}^{2n}$ such that
\begin{equation}
J(z_1) >1.
\label{eqPurity3}
\end{equation}
Let $z_0= \phi (z_1)$. By Lemma \ref{LemmaPurity1}, there exists $f \in L^2 (\mathbb{R}^{n})$ such that $|W f (z_0)|=2^n$. We thus have
\begin{equation}
|\left(U_{\phi} Wf \right)(z_1)| =|J(z_1) Wf (\phi(z_1))|> | Wf
(\phi(z_1))| = | Wf (z_0)| =2^n, \label{eqPurity4}
\end{equation}
which contradicts (\ref{eqPurity1}). Thus, we must have $J(z) \le 1$ for all $z\in \mathbb{R}^{2n}$.
\end{proof}

\begin{theorem}\label{TheoremCovariance}
Let $\mathcal{A}= \W^2, \W_+^2$ and let the operator $U_{\phi} \in
\mathcal{U}_{\A}$ be given by (\ref{eq2}). Then $U_{\phi}$ is a
map of the form
\begin{equation}
U_{\phi} : \A\to \A, \label{eq5.0}
\end{equation}
if and only if $\phi$ is given by:
\begin{equation}
\phi(z) =M z +a,
\label{eq5.1}
\end{equation}
with $a \in \mathbb{R}^{2n}$ and $M \in SpT(n)$.
\end{theorem}

\begin{proof}
Sufficiency is well known. To prove necessity, we start by showing that $\phi$ has to be of the form (\ref{eq5.1}) with $M \in Gl (2n ; \mathbb{R})$, that is a linear transformation followed by a translation. Indeed, let $f$ be a normalized Gaussian pure state. Then the corresponding Wigner function
\begin{equation}
W f(z)=2^n e^{- 2 \pi z \cdot \Sigma^{-1} z}
\label{eq5}
\end{equation}
is a positive function. Here $\frac{1}{4 \pi} \Sigma$ is the
covariance matrix of the Gaussian measure, which has to satisfy
$\Sigma \in Sp^+(n)$ \cite{Dias1,Littlejohn}, that is: it is a
real symmetric positive-definite symplectic $2n \times 2n$ matrix.
By assumption, under the transformation
\begin{equation}
Wf (z) \mapsto (U_{\phi} W f)(z) =  2^n J(z) e^{-2 \pi \phi (z)
\cdot \Sigma^{-1} \phi(z)}, \label{eq6}
\end{equation}
we obtain another Wigner function $Wf^{\prime}$ for $f^{\prime}
\in L^2 (\mathbb{R}^n)$. But since $U_{\phi} $ amounts to a
coordinate transformation, $U_{\phi} (W f)$ is also everywhere
nonnegative. By Hudson's Theorem (Theorem \ref{HudsonTheorem}),
$U_{\phi} (W f)$ must be some other Gaussian Wigner function,
i.e.:
\begin{equation}
(U_{\phi} W f) (z) = 2^n e^{-2 \pi (z-z_{\Sigma})\cdot
\Lambda_{\Sigma}^{-1} (z-z_{\Sigma})}, \label{eq7}
\end{equation}
with $\Lambda_{\Sigma} \in Sp^+ (n)$ and $z_{\Sigma} \in \mathbb{R}^{2n}$.

Notice that from equating (\ref{eq6}) and (\ref{eq7}), we must have
\begin{equation}
J(z) >0,
\label{eq7.1}
\end{equation}
for all $z \in \mathbb{R}^{2n}$. It is then safe to take the logarithm of (\ref{eq6}) and (\ref{eq7}). We conclude that $ \phi $ has the property that, for any $\Sigma \in Sp^+ (n)$, there exist $\Lambda_{\Sigma} \in Sp^+ (n)$ and $z_{\Sigma} \in \mathbb{R}^{2n}$, such that:
\begin{equation}
\ln (J(z)) - 2 \pi \phi (z) \cdot \Sigma^{-1}\phi (z) = - 2 \pi (z-z_{\Sigma}) \cdot \Lambda_{\Sigma}^{-1} (z-z_{\Sigma}) .
\label{eq8}
\end{equation}
If we take $z= z_{\Sigma}$ in the previous equation, we obtain
\begin{equation}
\ln (J(z_{\Sigma})) = 2 \pi \phi (z_{\Sigma}) \cdot \Sigma^{-1}\phi (z_{\Sigma}).
\label{eq8.A}
\end{equation}
From Lemma \ref{LemmaJacobian}, we must have $\ln (J(z_{\Sigma})) \le 0$. But, since the matrix $\Sigma^{-1}$ is positive-definite, this is possible if and only if $\phi (z_{\Sigma})=0$, that is
\begin{equation}
z_{\Sigma}= \phi^{-1} (0).
\label{eq8.B}
\end{equation}
Hence, $z_{\Sigma}$ does not depend on the choice of matrix $\Sigma$.

Next, by choosing judiciously $n(2n+1)$ points $z \in \mathbb{R}^{2n}$ in eq.(\ref{eq8}), we derive a linear system of $n(2n+1)$ independent equations for the entries of the matrix $\Lambda_{\Sigma}^{-1}$. By solving this system, we obtain the entries of $\Lambda_{\Sigma}^{-1}$ as polynomia (of degree one) of the entries of the matrix $\Sigma^{-1}$.

Let $\vec{\lambda} = (\lambda_1, \cdots, \lambda_n) \in (\mathbb{R}^+)^n$ and $D=diag(\lambda_1, \cdots, \lambda_n)$. Then
\begin{equation}
\vec{\lambda} \mapsto \Sigma^{-1} (\vec{\lambda}) = \left(
\begin{array}{c c}
D & 0\\
0 & D^{-1}
\end{array}
\right),
\label{eq8.1}
\end{equation}
defines a smooth mapping to $Sp^+(n)$. This, in turn, determines another smooth mapping $\Lambda_{\Sigma}^{-1} (\vec{\lambda})$ to $Sp^+(n)$ by (\ref{eq8}). We thus have:
\begin{equation}
\ln (J(z)) - 2 \pi \phi (z) \cdot \Sigma^{-1}(\vec{\lambda}) \phi (z) = - 2 \pi (z-\phi^{-1}(0)) \cdot \Lambda_{\Sigma}^{-1}(\vec{\lambda}) (z-\phi^{-1}(0)) .
\label{eq8.2}
\end{equation}
Differentiating the previous equation with respect to $\lambda_j$ yields:
\begin{equation}
\phi_j^2 - \lambda_j^{-2} \phi_{j+n}^2 = (z-\phi^{-1}(0)) \cdot \left[\frac{\partial}{\partial \lambda_j} \Lambda_{\Sigma}^{-1} (\vec{\lambda})\right] (z-\phi^{-1}(0)) .
\label{eq8.3}
\end{equation}
Since the limit $\lambda_j \to + \infty$ exists on the left-hand side, so does the one on the right-hand side, and we obtain:
\begin{equation}
\phi_j^2 = (z-\phi^{-1}(0)) \cdot B_j (z-\phi^{-1}(0)) ,
\label{eq8.4}
\end{equation}
for all $j=1, \cdots, n$ and where
\begin{equation}
B_j := \lim_{\lambda_j \to + \infty} \frac{\partial}{\partial \lambda_j} \Lambda_{\Sigma}^{-1} (\vec{\lambda})
\label{eq8.5}
\end{equation}
The limit is obviously performed component-wise, by regarding the
$2n \times 2n $ matrices as elements of $\mathbb{R}^{4n^2}$. If we multiply (\ref{eq8.3}) by $\lambda_j^2$ and send $\lambda_j
\downarrow 0$, we obtain
\begin{equation}
\phi_{j+n}^2 = (z-\phi^{-1}(0)) \cdot B_{j+n} (z-\phi^{-1}(0)) ,
\label{eq8.6}
\end{equation}
where this time
\begin{equation}
B_{j+n} := - \lim_{\lambda_j \downarrow 0 } \lambda_j^2 \frac{\partial}{\partial \lambda_j} \Lambda_{\Sigma}^{-1} (\vec{\lambda})
\label{eq8.7}
\end{equation}
Thus, for all practical purposes, $\phi_j^2$ is a polynomial of degree at most $2$ of the variables $z$ for $j=1, \cdots, 2n$.

Next, recall that $\Sigma^{-1} \in Sp^+ (n)$ if and only if there exists $A \in \textsf{sp} (n) \cap Sym (2n; \mathbb{R})$ such that $\Sigma^{-1}= e^A$ (see Proposition 2.18 in \cite{Birk}).

Thus, for $\epsilon  \ge 0$, let
\begin{equation}
\Sigma^{-1} ( \epsilon) = e^{\epsilon A} ,
\label{eq9}
\end{equation}
with $A \in \textsf{sp} (n) \cap Sym (2n; \mathbb{R})$. It follows that $\Sigma^{-1} ( \epsilon)$ describes a smooth path in $Sp^+(n)$, for $\epsilon \ge 0$. Again, that will induce another smooth path $\Lambda_{\Sigma}^{-1} (\epsilon)$  on $Sp^+(n)$ for the matrix appearing on the right-hand side of (\ref{eq8}).

If we substitute $\Sigma^{-1} ( \epsilon)$ in (\ref{eq8}), differentiate with respect to $\epsilon$ and send $\epsilon \downarrow 0$, we obtain:
\begin{equation}
\phi (z) \cdot A \phi (z) = (z-\phi^{-1}(0)) \cdot B_A (z-\phi^{-1}(0)) ,
\label{eq10}
\end{equation}
for all $z \in \mathbb{R}^{2n} $, and where
\begin{equation}
 B_A:=  \lim_{\epsilon \downarrow 0} \frac{d}{d \epsilon} \Lambda_{\Sigma}^{-1} (\epsilon).
\label{eq11}
\end{equation}
Recall that $A \in \textsf{sp} (n)$ if and only if
\begin{equation}
A \mathcal{J} + \mathcal{J} A^T =0. \label{eq12}
\end{equation}
Thus $A$ has to be of the form
\begin{equation}
A= \left(
\begin{array}{c c}
a & b\\
c & - a^T
\end{array}
\right)
\label{eq13}
\end{equation}
with $a,b,c $ real $ n \times n$ matrices and $b,c$ symmetric.
Hence, $A \in \textsf{sp} (n) \cap Sym (2n; \mathbb{R})$ if and
only if
\begin{equation}
A= \left(
\begin{array}{c c}
a & b\\
b & -a
\end{array}
\right),
\label{eq14}
\end{equation}
where $a,b$ are real symmetric $n \times n $ matrices.

Next, choose $a=0$ and $b= E^{(jj)}$, $j=1, \cdots, n$. Here $E^{(jj)}$ denotes the diagonal $n \times n$ matrix, whose $ jj$-th entry is one and all the remaining vanish. Substituting on the left-hand side of (\ref{eq10}), we obtain:
\begin{equation}
2 \phi_j \phi_{j+n}.
\label{eq14.1}
\end{equation}
Thus for all $j=1, \cdots,n$, we have shown that $\phi_j^2$, $\phi_{j+n}^2$ and $\phi_j \phi_{j+n}$ are polynomials of degree lower or equal to $2$. From Lemma \ref{LemmaPolynomials} (see Appendix), we have two possibilities.

\vspace{0.3 cm}
\noindent
(i) Either $\phi_j$ and $\phi_{j+n}$ are both polynomials of degree $\le 1$, or

\vspace{0.3 cm}
\noindent
(ii) $\phi_{j+n}$ and $\phi_j$ are proportional to each other.

\vspace{0.3 cm}
\noindent
Suppose that for some $j=1, \cdots,n$ possibility (ii) holds, that is: there exists a constant $\alpha_j \in \mathbb{R}$ such that $\phi_{j+n}= \alpha_j  \phi_j$ (or {\it vice-versa}). Then we conclude that the rows $j$ and $j+n$ of the matrix
\begin{equation}
\left(\frac{\partial \phi_i}{\partial z_k} \right)_{1 \le i,k \le 2n}
\label{eq14.1}
\end{equation}
are proportional to each other. Consequently, the Jacobian (\ref{eq2.1}) vanishes identically, and this possibility has to be ruled out. Altogether, we have shown that $\phi$ has to be of the form (\ref{eq5.1}) with $M \in Gl (2n ; \mathbb{R})$. Finally, from Theorem 1 (ii) of \cite{Dias1} it follows that $M \in SpT(n)$, which concludes the proof.
\end{proof}

The previous result is trivially generalized if we also admit tempered distributions:

\begin{corollary}\label{CorollaryCovariance}
Let $\mathcal{A}= \W^{\prime}$ and let $U_{\phi} \in
\mathcal{U}_{\A}$ be given by (\ref{eq2}). Then $U_{\phi}$ is a
map of the form:
\begin{equation}
U_{\phi}: \W^{\prime} \to \W^{\prime}, \label{eq14.1.1}
\end{equation}
if and only if:
\begin{equation}
\phi(z) =M z +a,
\label{eq14.2}
\end{equation}
with $a \in \mathbb{R}^{2n}$ and $M \in SpT(n)$.
\end{corollary}

\begin{proof}
Again, we start with the Gaussian (\ref{eq5}) and, upon the action
of $U_{\phi}$, we obtain (\ref{eq6}). We thus have a Wigner
distribution $U_{\phi} (W f) = W f^{\prime}$, for $f^{\prime} \in
\mathcal{S}^{\prime} (\mathbb{R}^n)$, which is everywhere
nonnegative. By Hudson's Theorem (Theorem \ref{HudsonTheorem}),
$f^{\prime}$ is either a pseudo-Gaussian or a Gaussian function.
Since $U_{\phi}$ acts as a coordinate transformation, then only
the latter hypothesis is possible and the rest of the proof
follows as before. This shows that the transformation $\phi$ is
the polynomial of (at most) degree one (\ref{eq14.2}) for some
matrix $M \in Gl(2n; \mathbb{R})$. Finally, since $U_{\phi} $ maps
$Wf$ with $f \in\mathcal{S} (\mathbb{R}^n)$ to some $W
f^{\prime}$, with $f^{\prime} \in \mathcal{S}^{\prime}
(\mathbb{R}^n)$, then again from Theorem 1 (ii) of \cite{Dias1},
it follows that $M \in SpT(n)$.
\end{proof}

Before we proceed let us make the following remark.

\begin{remark}\label{Remark1}
Since, in the proofs of the previous results, we basically used
the real Gaussian (\ref{eq5}) and applied Hudson's Theorem, we
also conclude that $U_{\phi}$ in (\ref{eq2}) is a map of the form
$U_{\phi}: \W_+^2 \to \W^{\prime}$, if and only if $\phi (z) = Mz
+ a$ with $a \in \mathbb{R}^{2n}$ and $M \in SpT (n)$.
\end{remark}

Finally, we prove the result for $U_{\phi}$ acting on the Wigner
functions of density matrices.

\begin{theorem}\label{TheoremPurity2}
Let $\mathcal{A}= \W_M$ and let $U_{\phi} \in \mathcal{U}_{\A}$ be
given by (\ref{eq2}). Then $U_{\phi}$ is a map of the form:
\begin{equation}
U_{\phi}: \W_M \to \W_M, \label{eqPurity3.1}
\end{equation}
if and only if $\phi$ is of the form:
\begin{equation}
\phi(z) =M z +a,
\label{eqPurity3.2}
\end{equation}
with $a \in \mathbb{R}^{2n}$ and $M \in SpT(n)$.
\end{theorem}

\begin{proof}
We start by showing that if $U_{\phi}$ is of the form
(\ref{eqPurity3.1}), then we must have
\begin{equation}
J(z)=1,
\label{eqPurity3}
\end{equation}
for all $z \in \mathbb{R}^{2n}$. From Lemma \ref{LemmaJacobian}, we already know that we must have $J(z) \le 1$, for all $z \in \mathbb{R}^{2n}$.

Next, suppose there exists $z_2 \in  \mathbb{R}^{2n}$ such that
\begin{equation}
J(z_2) <1.
\label{eqPurity5}
\end{equation}
Since $J(z)$ is a continuous function, there exists an open ball $B_{\epsilon} (z_2)$, for some $\epsilon>0$, such that
\begin{equation}
J(z) <1,
\label{eqPurity6}
\end{equation}
for all $z \in B_{\epsilon} (z_2)$.

Consider the Gaussian measure
\begin{equation}
\mathcal{G} (z) = 2^n e^{-2 \pi |z|^2}.
\label{eqPurity7}
\end{equation}
This is the Wigner function $\mathcal{G}=W f$ of the normalized
gaussian state
\begin{equation}
f(x) = 2^{n/4} e^{- \pi |x|^2}.
\label{eqPurity8}
\end{equation}
Next define
\begin{equation}
F(z) = N \mathcal{G} \left( \phi^{-1} (z) \right) ,
\label{eqPurity9}
\end{equation}
where
\begin{equation}
N = \left( \int_{\mathbb{R}^{2n}} \mathcal{G} \left( \phi^{-1} (z) \right) dz \right)^{-1}.
\label{eqPurity10}
\end{equation}
Clearly, $F$ is a real and normalized function. It also belongs to $L^2 (\mathbb{R}^{2n})$:
\begin{equation}
\begin{array}{c}
\int_{\mathbb{R}^{2n}} |F(z)|^2 dz = N^2 \int_{\mathbb{R}^{2n}} | \mathcal{G} \left(\phi^{-1} (z) \right)|^2 dz = \\
\\
= N^2 \int_{\mathbb{R}^{2n}} | \mathcal{G} (u)|^2 J(u) du \le N^2 || \mathcal{G}||_{L^2 (\mathbb{R}^{2n})}^2,
\end{array}
\label{eqPurity11}
\end{equation}
where we performed the substitution $u=  \phi^{-1} (z)$ and used the fact that $J(u) \le 1$ everywhere.

However, $F$ cannot be a Wigner function, as we now show. Indeed, from (\ref{eqPurity6}) and the fact that $\mathcal{G}$ is everywhere nonnegative:
\begin{equation}
N = \left( \int_{\mathbb{R}^{2n}} \mathcal{G} (u) J(u) du \right)^{-1} > \left( \int_{\mathbb{R}^{2n}} \mathcal{G} (u) du \right)^{-1}=1.
\label{eqPurity12}
\end{equation}
Let $z_3 = \phi(0)$. Then:
\begin{equation}
F(z_3) = N \mathcal{G} \left(\phi^{-1} (z_3) \right) =  N \mathcal{G} (0) =N 2^n > 2^n,
\label{eqPurity13}
\end{equation}
which contradicts (\ref{eqPurity1}). Hence, $F$ is not a Wigner function. From Theorem \ref{TheoremWignerFunctions}, we conclude that there exists at least one Wigner function $W f$ such that
\begin{equation}
\int_{\mathbb{R}^{2n}} F(z) W f (z) dz < 0.
\label{eqPurity14}
\end{equation}
On the other hand:
\begin{equation}
\begin{array}{c}
\int_{\mathbb{R}^{2n}} F(z) W f (z) dz  = N \int_{\mathbb{R}^{2n}} \mathcal{G} \left(\phi^{-1} (z) \right) W f (z) dz = \\
\\
=N \int_{\mathbb{R}^{2n}} \mathcal{G} (u) W f \left(\phi(u)
\right) J(u) du = N \int_{\mathbb{R}^{2n}} \mathcal{G} (u)
\left(U_{\phi} W f \right) (u)  du.
\end{array}
\label{eqPurity15}
\end{equation}
and since $\mathcal{G}$ is a Wigner function, it follows from
(\ref{eqPurity14}) that $U_{\phi} (W f)$ cannot be a Wigner
function. Hence (\ref{eqPurity5}) cannot hold.

So, if $J(z)=1$ everywhere, then $U_{\phi}$ preserves the purity
$||U_{\phi}(W \rho)||_{L^2 (\mathbb{R}^{2n})}= ||W \rho||_{L^2
(\mathbb{R}^{2n})}$. Hence, $U_{\phi}$ maps pure states to pure
states and the rest of the proof follows from Theorem
\ref{TheoremCovariance}.
\end{proof}

\begin{remark}\label{Remark2}
It also follows from the proof of the previous Theorem that
$U_{\phi}$ is a map of the form $U_{\phi}: \W_+^2 \to \W_M$ if and
only if $\phi (z) = Mz +a$, with $ a \in \mathbb{R}^{2n}$ and $M
\in SpT(n)$.
\end{remark}

\section*{Appendix}

In this Appendix we prove a lemma for polynomials on $\RE^n$. It
is a simple result that looks quite natural, but we were unable to
find it in the literature. Since it plays an important part in the
derivation of Theorem \ref{TheoremCovariance}, we will prove it
here for completeness.

Let us define the spaces $\P_k$, $k\in \mathbb{N}_0$, of real polynomials on $\RE^n$, of degree at most $k$: $f
\in \P_k$ iff it is of the form
\begin{equation}
f: \RE^n \longrightarrow \RE; \quad f(x_1,...,x_n)= \sum_{|\alpha|=0}^k C_\alpha x_1^{\alpha_1}\cdot \cdot \cdot
x_n^{\alpha_n}, \label{appendix1}
\end{equation}
where $\alpha =(\alpha_1, \cdots, \alpha_n) \in \mathbb{N}^n$ is a multiindex, $|\alpha|=\sum_{i=1}^n \alpha_i$, and $C_\alpha \in \RE$. The highest value of $|\alpha|$ for which
$C_\alpha \not=0$ is the {\it degree} of $f$, denoted by deg $(f)$. If $f \in \P_k$ then, of course, deg $(f)
\le k$.

Let us also define the set $\widetilde\P_2$ of second degree polynomials $f \in \P_2$ for which there is $f_1
\in \P_1 \backslash \P_0$ such that $f =f_1^2$.

Finally, recall that every $f \in \P_k$ can be factorized in the following way: $f=f_1 \cdot \cdot \cdot f_s$
$(s \le k)$, where all $f_j \in \P_k$, $j=1,...,s$, are irreducible polynomials (i.e. cannot be factorized into products of
lower degree polynomials); deg $(f)= \sum_{j=1}^s$ deg $(f_j)$ and the factorization is unique up to
multiplications of the factors by real numbers.

We now state our result:

\begin{lemma}\label{LemmaPolynomials}
Let $f,g: \mathbb{R}^n \to \mathbb{R}$ be two continuous functions, such that

$f^2,  g^2,  f \cdot g \in \P_2$.
Then one of the following two possibilities must hold:

\vspace{0.3 cm} \noindent {\bf A)} $f,g \in \P_1$, or

\vspace{0.3 cm} \noindent {\bf B)} there exist $(\lambda, \mu) \in \mathbb{R}^2 \backslash \{(0,0)\}$ such that
$\lambda f + \mu g =0$.
\end{lemma}

\begin{proof}

Let us define $F:=f^2$, $G:=g^2$ and $H:=f \cdot g$. Notice that $F,G$ can belong to $\P_0$ or to $\P_2$, but not
to $\P_1 \backslash \P_0$ (in which case they would not be everywhere nonnegative). We have several cases:

\vspace{0.3 cm} \noindent {\bf Case 1:} $F,G \in \P_0$.

This case is trivial: $f,g \in \P_0 \subset \P_1$ and so {\bf A} holds.

\vspace{0.3 cm} \noindent {\bf Case 2:} $G \in \P_0$ and $F \in \P_2 \backslash \P_1$ or vice-versa.

We have $g \in \P_0$ and since $H=g \cdot f \in \P_2$, we also have $f \in \P_2$. Since $F \in \P_2 \backslash \P_1$, this implies
that $f \in \P_1$ and so {\bf A} is satisfied.

\vspace{0.3 cm} \noindent {\bf Case 3:} $F,G \in \P_2 \backslash \P_1$. We have two sub-cases:

\vspace{0.1 cm} {\bf Sub-case 3.1:} $F \in \widetilde\P_2$ (or $G \in \widetilde\P_2$).

Since $F \in \widetilde\P_2$, there is $h \in \P_1 \backslash \P_0$ such that $h^2=F$. Hence, $f^2=h^2$. The only continuous
solutions of this equation are $f=\pm h$ and $f=\pm |h|$.

Next we notice that $H \in \P_2 \backslash \P_0$, and so one of the following possibilities holds:
\begin{equation}
(i) \, H=h_1^2 \quad , \quad (ii)\, H=h_2 \cdot h_3 \quad , \quad (iii)\, H \mbox{ is not factorisable,}  \label{Appendix2}
\end{equation}
where $h_1,h_2,h_3 \in \P_1 \backslash \P_0$, and $h_2$ is not proportional to $h_3$. Since $F \cdot G =H^2$, we have for $G$
\begin{equation}
(i) \, h^2 \cdot G=h_1^4 \quad , \quad (ii) \, h^2 \cdot G=h_2^2 \cdot h_3^2 \quad , \quad (iii) \, h^2 \cdot
G=H^2. \label{Appendix3}
\end{equation}
The solutions are
\begin{equation}
(i) \, \left\{ \begin{array}{l} G= k h_1^2 \\
h^2=h_1^2/k \end{array} \right. \quad , \quad (ii) \, \left\{ \begin{array}{l} G= k h_3^2 \\
h^2=h_2^2/k \end{array} \right. \quad \mbox{or} \quad
\left\{ \begin{array}{l} G= k h_2^2 \\
h^2=h_3^2/k \end{array} \right. \label{Appendix4}
\end{equation}
where $k \in \RE_+$ and we used the fact that $G$ and $h^2$ are non-negative polynomials. The case (iii) of
equation (\ref{Appendix3}) has no solutions (i.e. there is no $G \in \P_2$ for which (iii) holds).

We conclude from (\ref{Appendix4}) that $G \in \widetilde P_2$. Moreover, we get (using the continuity of $f$
and $g$):
\begin{equation}
(i) \, \left\{ \begin{array}{l} g= \pm \sqrt{k} h_1 \\
f=\pm h_1/\sqrt{k} \end{array} \right. \quad \mbox{or} \quad
\left\{ \begin{array}{l} g= \pm \sqrt{k} |h_1| \\
f=\pm |h_1|/\sqrt{k} \end{array} \right. , \label{Appendix5}
\end{equation}
\begin{equation}
(ii) \, \left\{ \begin{array}{l} g= \pm \sqrt{k} h_3 \\
f=\pm h_2/\sqrt{k} \end{array} \right. \quad \mbox{or} \quad
\left\{ \begin{array}{l} g= \pm \sqrt{k} h_2 \\
f=\pm h_3/\sqrt{k} \end{array} \right. \label{Appendix6}
\end{equation}
In the case (ii), we cannot have $f \propto |h_2|$, nor $g \propto |h_3|$ because then $f \cdot g \notin \P_2$.

We conclude that in case (i) {\bf B} holds, whereas case (ii) implies {\bf A}.

\vspace{0.1 cm} {\bf Sub-case 3.2:} $F,G \in \P_2 \backslash \widetilde\P_2$.

In this case $F$ and $G$ are irreducible (notice that if e.g. $F=f_1 \cdot f_2$, with $f_1,f_2 \in \P_1$, and $f_1$ not proportional to $f_2$, then $F$ would not be non-negative). Since $F \cdot G =H^2$ we get
\begin{equation}
F= k | H | \quad \mbox{and} \quad G= |H|/k \quad , \quad k \in \RE_+  \label{Appendix6}
\end{equation}
Notice that $|H| \in \P_2$. We conclude that $H \in \P_2 \backslash \widetilde\P_2$ and
\begin{equation}
f=\pm \sqrt{k} \sqrt{|H|} \quad \mbox{and} \quad g=\pm \sqrt{|H|}/\sqrt{k} \quad , \quad k \in \RE_+
\label{Appendix7}
\end{equation}
Hence, {\bf B} holds and we have concluded the proof.
\end{proof}

\section*{Acknowledgements}

The work of N.C. Dias and J.N. Prata is supported by the
Portuguese Science Foundation (FCT) grant PTDC/MAT-CAL/4334/2014.

***************************************************************

\noindent \textbf{Author's addresses:}

\noindent \textbf{\ Nuno Costa Dias}, \textbf{Jo\~{a}o Nuno Prata
}:  Grupo de F\'{\i}sica
Matem\'{a}tica, Departamento de Matem\'atica, Faculdade de Ci\^encias, Universidade de Lisboa, Campo Grande, Edif\'{\i}cio C6, 1749-016 Lisboa, Portugal, and Escola Superior N\'autica Infante D. Henrique. Av.
Eng. Bonneville Franco, 2770-058 Pa\c{c}o d'Arcos, Portugal

\end{document}